\documentclass{article}

\usepackage[T1]{fontenc}
\usepackage[utf8]{inputenc}
\usepackage[shortlabels]{enumitem}
\usepackage{bussproofs}
\usepackage[pdfusetitle,hidelinks]{hyperref}
\usepackage[backend=biber,
maxcitenames=4,
maxbibnames=4,
minalphanames=4,
maxalphanames=4,
style=alphabetic]
{biblatex}
\usepackage{orcidlink}

\usepackage{amsthm}
\usepackage{amssymb}
\usepackage{amsmath}

\usepackage{graphicx}

\usepackage{cleveref}

\newtheorem{definition}{Definition}
\newtheorem{theorem}{Theorem}
\newtheorem{lemma}[theorem]{Lemma}
\newtheorem{fact}[theorem]{Fact}
\newtheorem{claim}[theorem]{Claim}
\newtheorem{corollary}[theorem]{Corollary}

\begin{document}

\title{On Protocols for Monotone Feasible Interpolation}

\author{Lukáš Folwarczný \orcidlink{0000-0002-3020-6443}\\
\small Institute of Mathematics of the Czech Academy of Sciences\\
\small Prague, Czech Republic\\[.5em]
\small Computer Science Institute of Charles University\\
\small Prague, Czech Republic\\
\small \tt \href{mailto:folwarczny@math.cas.cz}{folwarczny@math.cas.cz}}

\maketitle

\begin{abstract}
Feasible interpolation is a general technique for proving proof
complexity lower bounds. The monotone version of the technique converts, in its basic variant, lower bounds for monotone Boolean circuits separating two NP-sets to proof complexity lower bounds. In a generalized version of the technique, dag-like communication protocols are used instead of monotone Boolean circuits. We study three kinds of protocols and compare their strength.

Our results establish the following relationships in the sense of polynomial reducibility: Protocols with equality are at least as strong as protocols with inequality and protocols with equality have the same strength as protocols with a conjunction of two inequalities. Exponential lower bounds for protocols with inequality are known. Obtaining lower bounds for protocols with equality would immediately imply lower bounds for resolution with parities (R(LIN)).
\end{abstract}

\section{Introduction}
\label{sec:introduction}

Studying the strength of various propositional proof systems is the essence of proof complexity. Dag-like communication protocols come into play as a~tool to translate, via a technique called feasible interpolation, the task of proving lower bounds on the length of proofs into the realm of communication complexity. This paper studies variants of protocols which seem to have a~potential for producing new proof complexity lower bounds. Most importantly, lower bounds for a certain type of protocols, called protocols with equality in this paper, would immediately translate into lower bounds for the proof system resolution with parities (R(LIN)).

\subsection{Motivation and context}

Propositional proof system, in the sense of Cook and Reckhow \cite{cookreckhow}, is a~sound and complete system for propositional tautologies with proofs verifiable in polynomial time.  Examples of propositional proof systems include resolution, Frege systems, sequent calculus and cutting planes. Proof complexity is then the field studying the strength of various propositional proof systems; proof complexity is tightly connected with computational complexity (one of the fundamental open problems is equivalent to the question whether NP = coNP) and logic. A~reference for this field is the book by Krajíček~\cite{krajicekproofcomplexity}. A~lower bound is a~theorem stating that for a~certain proof system~$P$ and a~sequence of tautologies $\left\{ \phi_n \right\}$ the shortest length of a~proof of $\phi_n$ in $P$ is bounded from below by a~certain function of the size of $\phi_n$. Exponential lower bounds have been proven for various proof systems, Resolution, constant depth Frege systems, Cutting Planes, Polynomial Calculus (see Kraj\'{i}\v{c}ek's monograph~\cite{krajicekproofcomplexity}), but for other proof systems studied in the literature this is still an open problem. 

The feasible interpolation method was invented by Krajíček (idea formulated in \cite{krajicek1994}, applied in \cite{krajicek-feasibleinterpolation}). In the basic setup, feasible interpolation reduces the task of proving a~lower bound for a~proof system~$P$ to proving a~lower bound for Boolean circuits separating two NP-sets. In the case of monotone feasible interpolation, lower bounds for monotone Boolean circuits are enough. However, lower bounds for different objects than monotone Boolean circuits may be used as well. Limits of monotone feasible interpolation by monotone Boolean circuits were already considered in the aforementioned paper by Krajíček~\cite{krajicek-feasibleinterpolation}, Section~9. However, general limits of monotone feasible interpolation are not known. The first result when monotone feasible interpolation was used with another computational model than Boolean circuits is due to Pudlák~\cite{pudlak-interpol}: He defined a~generalization of monotone Boolean circuits called monotone real circuits (which were later proved by Rosenbloom~\cite{rosenbloom} to be strictly stronger than monotone Boolean circuits) and proved lower bounds for this model which lead, via monotone feasible interpolation, to lower bounds for the cutting planes proof system.

The history of the (dag-like communication) protocols considered in this paper starts with a paper by Karchmer and Wigderson \cite{karchmerwigderson} where the following theorem is proved:
For a~Boolean function~$f$, the minimum depth of a~Boolean circuit computing $f$ is equal to the communication complexity (that is the depth of a~protocol) of a~certain relation defined for~$f$. There is also a~monotone version of the relation for monotone circuits.

In classical communication complexity, the measure of protocols is their depth and hence one can without loss of generality consider protocols with the underlying graph being a~tree. Dag-like communication protocols, that is protocols with the underlying graph being a~directed acyclic graph, go back to Razborov~\cite{razborov-unprovability}. Razborov, inspired by Karchmer and Wigderson, used the size of certain protocols to characterize the size of circuits. See Pudlák~\cite{pudlak-extractingcomputations} or Sokolov~\cite{sokolov-daglike} for a~survey of this topic. We give more details on the origin of the protocols in \Cref{sec:protocols}. More recently, new types of protocols have been introduced~\cite{garggooskamathsokolov}.

It is now a well known open problem in proof complexity to prove lower bounds for the system which we denote, as in the Krajíček's book, by R(LIN). This problem may be seen as a step in solving another open problem which is  to prove lower bounds for the so called $AC^0[p]$-Frege proof systems. The system R(LIN) was introduced by Itsykson and Sokolov~\cite{itsyksonsokolov} who also prove exponential lower bounds for the tree-like version.
Itsykson and Sokolov call the system $\mathrm{Res}(\oplus)$ (or Res-Lin in a preliminary version of the paper). The system was inspired by the system R(lin) introduced by Raz and Tzameret~\cite{raztzameret}.

Lower bounds for protocols called \emph{protocols with equality} in this paper would directly lead to lower bounds for R(LIN). A~different approach which could work for this proof system is randomized feasible interpolation due to Krajíček~\cite{krajicek-random-fi}, studied also by Krajíček and Oliveira~\cite{krajicek-oliveira}.

\subsection{Our contribution}

In this section, all protocols have constant degree. Our contribution is the comparison of strength of protocols: Protocols with equality are at least as strong as protocols  with inequality (in the sense of polynomial reducibility). Furthermore, we establish that protocols with a conjunction of a constant number of inequalities have the same strength as protocols with equality if the constant is at least 2. Details are given in \Cref{sec:results}.

\subsection{Open problems}

Research into the dag-like communication complexity of KW games is motivated by the quest of proving lower bounds on proof systems that are stronger than those for which we already have lower bounds. Since R(LIN) is just on the border of what we can and cannot prove and lower bounds on protocols with equality would give us such a lower bound, the main problem is to prove a superpolynomial lower bound on such protocols.

Another motivation is the general question of how far we can get with feasible interpolation. This question is mentioned by Razborov~\cite{razborov-sigact} in his survey as one of the main challenges in proof complexity. Lower bounds based on monotone interpolation reduce the problem of proving the lower bound on the lengths of proofs to proving a lower bound on some monotone circuits (monotone Boolean circuits, monotone Real circuits, monotone span programs). Researchers in the area of proof complexity agree on the idea that in order to make progress, we have to study communication protocols instead of circuit models, because there may be no circuit type for some proof system. It seems that this is the case of R(LIN). Therefore proving lower bounds on protocols with equality seems to be the most viable approach to proving lower bounds on R(LIN) proofs.

Since we do not have lower bounds for protocols with equality, but we do have exponential lower bounds for protocols with inequality, we expect the former type of protocols to be stronger than the latter. However we do not have any partial Boolean function that would provide such a separation. Finding such a function would also shed more light on the problem of proving lower bounds on protocols with equality.

\subsection{Paper outline}

In \Cref{sec:protocols}, we define all the protocols concerned in this paper and explain their origin.
In \Cref{sec:results}, we state our results, discuss them and prove a part of
them.  In \Cref{sec:proofmainthm}, we prove the main theorem.
In \Cref{sec:applications}, we show how lower bounds for protocols would translate into lower bounds for the proof systems R(LIN).

\section{Protocols}
\label{sec:protocols}

\subsection{General protocols}

Karchmer and Wigderson~\cite{karchmerwigderson} used classical tree-like (i.e.\ the
underlying graph is a~tree)
communication protocols and their depth to characterize the depth of circuits.
Dag-like protocols (i.e.\ the underlying graph is a~directed acyclic graph) considered in
this paper go back to Razborov~\cite{razborov-unprovability} who was considering the size
of protocols to characterize the size of circuits. An explicit definition of these protocols was given by
Krajíček~\cite{krajicek-feasibleinterpolation}. Our definition is a~generalization of the
definition by Hrubeš and Pudlák~\cite{hrubespudlak-monotonecircuits}.
Unlike in~\cite{krajicek-feasibleinterpolation}, in this paper only the monotone version is
defined and the strategy function is omitted. 

The task of the monotone Karchmer-Wigderson game (the monotone KW game) for a~given partial monotone Boolean function
$f$ is: given $x\in f^{-1}(0)$ and $y \in f^{-1}(1)$, find an index $i$ such that $x_i
= 0 \wedge y_i = 1$.
(Observe that there is always at least one such index $i$.)

\begin{definition}
\label{def:pgeneral}
Let $n \geq 1$ be a~natural number. A \emph{protocol of degree $d$ with
feasibility relation} is a~directed acyclic graph $G = (V, E)$ and a~relation $F
\subseteq \{0,1\}^n \times \{0,1\}^n \times V$, such that
\begin{enumerate}[(i)]
\item $G$ has one source $v_0$ (a~node of in-degree zero) and the out-degree of every
vertex is at most $d$,
\item for every sink $\ell$ (a~node of out-degree zero), there exists an index
$i$ such that for every $x, y \in \{0,1\}^n$ it holds $(x, y, \ell) \in F$ iff $x_i = 0$ and $y_i = 1$.
\end{enumerate}
Let $f$ be a~partial monotone Boolean function in $n$~variables. We say that the
protocol \emph{solves the monotone KW game for $f$} (or simply \emph{solves} $f$), if for
every $x \in f^{-1}(0)$ and $y \in f^{-1}(1)$,
\begin{enumerate}[(a)]
\item $(x, y, v_0) \in F$,
\item for every $v \in V$ with $p \geq 1$ children $u_1 , \dots , u_p$, if
$(x, y, v) \in F$ then there exists $u_i$ with $(x,y,u_i)\in F$.
\end{enumerate}
The size of a~protocol is the number of vertices.
\end{definition}

We say that a~vertex $v$ is \emph{feasible for
$x$, $y$} if $(x,y,v)\in F$. By definition, the source is feasible for any $x \in
f^{-1}(0)$, $y \in f^{-1}(1)$.  The
protocol solves the monotone KW game in the following sense: Given $x \in f^{-1}(0)$ and
$y \in f^{-1}(1)$ and any vertex $v$ which is feasible for $x$, $y$ (it is crucial that this holds
for any feasible vertex, not just the source), we can find the solution to the KW game by
traversing the graph via feasible vertices down to sinks which give us the solution.

Another general definition was given by Garg et al.~\cite{garggooskamathsokolov}, Section 2.1. In their
definition, degree is fixed to 2 and general search problems are considered. If we
restrict our definition to degree 2 and their definition to the monotone KW game, the definitions
become equivalent.

\subsection{Protocols with inequality and equality}

The definition of what we call a~protocol with inequality was independently
introduced by Hrubeš and Pudlák~\cite{hrubespudlak-monotonecircuits} and
Sokolov~\cite{sokolov-daglike} (Sokolov considers only protocols of degree~2).
Before that, Krajíček~\cite{krajicek-interpolationbyagame} considered a different type of
protocols with inequality.

\begin{definition}
\label{def:pinequality}
A~\emph{protocol of degree $d$ with inequality} is a~protocol of degree $d$ with feasibility relation such that the relation $F$ may be expressed as follows: For each vertex $v \in V$, there is a~pair of functions $q_v, r_v \colon \{0, 1\}^n \rightarrow \mathbb{R}$
such that for all $x, y \in \{0, 1\}^n$, it holds  $(x, y, v) \in F$ iff $q_v(x) < r_v(y)$.
\end{definition}

Hrubeš and Pudlák~\cite{hrubespudlak-monotonecircuits} proved that the size of the minimum
protocol of degree~$d$ with inequality solving $f$ and the size of the minimum $d$-ary
monotone real circuit computing $f$ are equal (definition of monotone real circuits may be
found in \cite{pudlak-interpol} or \cite{hrubespudlak-monotonecircuits}). This implies
that the exponential lower bounds due to Pudlák~\cite{pudlak-interpol} and Haken and Cook~\cite{cook-haken} for monotone real circuits hold also for protocols with inequality.
Replacing inequality with equality or a~conjunction of inequalities, we obtain the following two definitions:

\begin{definition}
\label{def:pequality}
A~\emph{protocol of degree $d$ with equality} is a~protocol of degree $d$ with feasibility relation such that the relation $F$ may be expressed as follows: For each vertex $v \in V$ there is a~pair of functions $q_v, r_v \colon \{0, 1\}^n \rightarrow \mathbb{R}$
such that for all $x, y \in \{0,1\}^n$ it holds $(x, y, v) \in F$ iff $q_v(x) = r_v(y)$.
\end{definition}

\begin{definition}
\label{def:pconjunction}
A~\emph{protocol of degree $d$ with a~conjunction of $c$~inequalities}
is a protocol of degree $d$ with
feasibility relation $F$ such that the relation $F$ may be expressed as follows: For each vertex
$v \in V$, there are $c$~pairs of functions $q^{j}_v, r^{j}_v \colon \{0, 1\}^n
\rightarrow \mathbb{R}$ for $j \in [c]$. The functions satisfy for all $x, y \in \{0,1\}^n$
$(x, y, v) \in F$ iff $q^1_v(x) < r^1_v(y) \wedge \cdots \wedge q^{c}_v(x) < r^{c}_v(y)$.
\end{definition}

All three types of the protocols defined in this subsection are mentioned by Garg et al.
\cite{garggooskamathsokolov} who name the protocols after the shape of the feasible sets
in the protocol. For a~given $v$, the feasible set for $v$ is the set of all pairs
$(x, y)$ such that $(x,y,v) \in F$. For protocols with inequality, the feasible sets are
combinatorial triangles (definition may be found in \cite{garggooskamathsokolov}). For
protocols with equality, the feasible sets are block-diagonal. And for protocols with
a~conjunction of $c$~inequalities, the feasible sets are intersections of $c$~combinatorial
triangles.

\section{Results}
\label{sec:results}

We use the $\mathcal{O}$-notation for functions with several parameters. The precise meaning is the following:
$$g \in \mathcal{O}(f) \Leftrightarrow \exists c>0 \forall n_1 \in \mathbb{N} \dots \forall n_k \in \mathbb{N}: g(n_1, \dots, n_k) \leq cf(n_1, \dots, n_k) + c$$

Our first theorem compares protocols with inequality and protocols with equality.

\begin{theorem}
\label{thm:c1}
Let $P$ be a protocol of degree $d$ with inequality solving an $n$-bit partial
monotone Boolean function $f$. If the size of $P$ is $s$, then there exists a protocol of degree 2 with equality solving $f$ whose size is $\mathcal{O}(sn^{d+2})$.
\end{theorem}

A more general theorem compares protocols with a conjunction of $c$~inequalities and protocols with equality.

\begin{theorem}[main]
\label{thm:main}
Let $P$ be a protocol of degree $d$ with a~conjunction of $c$ inequalities solving an $n$-bit partial monotone Boolean function $f$. If the size of $P$ is $s$, then there exists a protocol of degree~2 with equality solving $f$ whose size is $\mathcal{O}(sn^{2cd+c-1})$.
\end{theorem}

Hrubeš and Pudlák \cite{hrubespudlak-monotonecircuits} prove how to reduce the degree of protocols with inequality and correspondingly monotone real circuits. (We use their result as a part of the proof of \Cref{thm:c1}.)

\begin{lemma}[\cite{hrubespudlak-monotonecircuits}, Corollary 6 (ii)]
\label{l:hp}
Let $P$ be a~protocol of degree $d$ with inequality solving an $n$-bit partial monotone Boolean function $f$. If the size of $P$ is $s$, then there exists a~protocol of degree 2 with inequality solving $f$ whose size is $\mathcal{O}(sn^{d-2})$.
\end{lemma}

Our results enable a similar reduction for protocols with equality.

\begin{corollary}
\label{cor:degreereduction}
Let $P$ be a~protocol of degree $d$ with equality solving an $n$-bit partial monotone Boolean function $f$. If the size of $P$ is $s$, then there exists a~protocol of degree~2 with equality solving $f$ whose size is $\mathcal{O}(sn^{4d+1})$.
\end{corollary}

\subsection{Discussion}

We consider the theorems to be most relevant for the cases when $d = \mathcal{O}(1)$ and $c = \mathcal{O}(1)$.
In this case, we interpret the results as comparison of the strength of protocols in the sense of polynomial reducibility.

We cannot expect interesting results for very large degree because of the following simple fact:

\begin{fact}
\label{fact:degreen}
For every $n$-bit partial monotone Boolean function $f$, there is a~protocol of degree~$n$ with inequality solving $f$ whose size is $n+1$.
\end{fact}

Similarly, we cannot expect interesting results if the conjunctions are large.

\begin{fact}
\label{fact:ninequalities}
For every $n$-bit partial monotone Boolean function $f$, there is a~protocol of degree~2  with a~conjunction of $n-2$ inequalities solving $f$ whose size is $2n-1$.
\end{fact}

It is easy to see that protocols with a conjunction of two (or more than two) equalities are at least as strong as protocols with equality.

\begin{fact}
\label{fact:equalitybytwoinequalities}
Let $P$ be a~protocol of degree~$d$ with equality solving an $n$-bit partial
monotone Boolean function $f$. If the size of $P$ is $s$, then there exists
a~protocol of degree $d$ with a conjunction of two inequalities solving $f$ whose size is $s$.
\end{fact}

\subsection{Proofs}

We postpone the proof of \Cref{thm:main} to \Cref{sec:proofmainthm}.

\begin{proof}[Proof of \Cref{thm:c1}]
We use \Cref{l:hp} to convert $P$ into a protocol $P'$ of degree 2 with equality. The size is $\mathcal{O}(sn^{d-2})$. We then use \Cref{thm:main} with $c = 1$ and $d = 2$ to convert $P'$ into a protocol of size $\mathcal{O}(sn^{d-2} \cdot n^4) = \mathcal{O}(sn^{d+2})$. 
\end{proof}

\begin{proof}[Proof of \Cref{cor:degreereduction}]
We use \Cref{fact:equalitybytwoinequalities} to convert $P$ into a~protocol $P'$ of degree~$d$ with a~conjunction of two inequalities whose size is $s$. Using \Cref{thm:main} with $c = 2$ for $P'$, we obtain a~protocol of degree~2 with equality solving $f$ whose size is $\mathcal{O}(sn^{4d + 1})$.
\end{proof}

\section{Proof of the main theorem}
\label{sec:proofmainthm}

We prove \Cref{thm:main} in this section.

\subsection{Structure of the proof}

Let $G = (V, E)$ be the underlying graph of $P$ and let $q^j_v$, $r^j_v$ for $j \in [c]$ be the functions associated with $v \in V$. Observe that for fixed inner  $v \in V$ and $j \in [c]$, it holds
$$\left|\{ q^j_v(x) \mid x \in f^{-1}(0) \} \cup \{ r^j_v(y) \mid y \in f^{-1}(1) \}\right| \leq \left| f^{-1}(0) \cup f^{-1}(1) \right| \leq 2^n.$$
As the only thing that matters is the relative order of the values, we can w.l.o.g.\ assume that $q^j_v(x), r^j_v(y) \in \{0, \dots, 2^n - 1\}$. (In fact, the functions $q^j_v(x)$ and $r^j_v(y)$ are defined for any $x, y \in \{0,1\}^n$. However, we ignore the values for $x \notin f^{-1}(0)$ and $y \notin f^{-1}(1)$.) We treat $q^j_v(x)$ and $r^j_v(y)$ as $n$-bit numbers. The $i$-th most significant
bit of $q^j_v(x)$ (resp. $r^j_v(y)$) is denoted by $q^j_v(x)[i]$ (resp. $r^j_v(y)[i]$). That is
$$q^j_v(x) = \sum_{i=1}^{n} 2^{n-i}q^j_v(x)[i] \quad\text{and}\quad r^j_v(y) = \sum_{i=1}^{n} 2^{n-i}r^j_v(y)[i].$$

The proof based on the fact that one can express inequality of $n$-bit numbers as one of $n$ particular equalities. Denote $a = q^j_v(x)$ and $b = r^j_v(y)$. Then
$$
a < b \Leftrightarrow \exists i \in [n] \left(a[1] = b[1] \wedge \cdots \wedge a[i-1] = b[i-1] \wedge a[i] = 0 \wedge 1 = b[i] \right).
$$
The conjunction can be written as a single equality as
\begin{equation}\label{ineq}
a < b \Leftrightarrow \exists i \in [n] \left(a[1]\dots a[i-1]a[i]1 = b[1]\dots b[i-1] 0 b[i]\right).
\end{equation}
An important feature of the above expression is that the left-hand side of each equality is a function of only $x$ (it does not depend on $y$) and the right-hand side of each equality is a function of only $y$.

We can similarly express the conjunction of $c$ inequalities as one of $n^c$ equalities. Consider an inner vertex $v \in V$ and define:
\begin{align*}
\chi^{j,(i,=)}_v &\equiv q_v^j(x)[i] = r_v^j(y)[i]\\
\chi^{j,(i,\neq)}_v &\equiv q_v^j(x)[i] = 1-r_v^j(y)[i]\\
\chi^{j,(i,<)}_v &\equiv q_v^j(x)[i] = 0 \wedge 1=r_v^j(y)[i]\\
\chi^{j,(i,>)}_v &\equiv q_v^j(x)[i] = 1 \wedge 0=r_v^j(y)[i]
\end{align*}
 (we use the symbol $\equiv$ for definitions of formulas).
We can then write
$$\bigwedge_{j=1}^c \left(q^j_v(x) < r^j_v(y)\right) \Leftrightarrow \exists (i_1, \dots, i_c) \in [n]^c \bigwedge_{j=1}^c \left(\bigwedge_{k=1}^{i_j-1}\chi^{j,(k,=)}_v \wedge \chi^{j,(i_j,<)}_v\right).$$ The expression bounded by the existential quantifier is a conjunction of equalities of bits such that in each equality the left-hand side depends only on $x$ and the right-hand side depends only on $y$. Therefore, we can again rewrite the expression into a single equality. This is the desired expression of a conjunction of $c$ inequalities as one of $n^c$ equalities (the particular form of the equalities will be of importance, too).

We introduce the following convention: Vertices in the protocol with equality $P'$ will be labeled with conjunctions of equalities of bits such that the left-hand side of each equality is a function of $x$ and the right-hand side of each equality is a function of $y$. If a vertex $v$ in $P'$ is labeled with $\psi$ defined as $$b_1(x) = b'_1(y) \wedge \cdots \wedge b_\ell(x) = b'_\ell(y),$$ then the functions associated with $v$ in $P'$ are
$$q_v(x) := b_1(x)b_2(x)\ldots b_\ell(x) \quad\text{and}\quad r_v(y) := b'_1(y)b'_2(y)\ldots b'_\ell(y).$$ Hence, $v$ labeled with $\psi$ is feasible in the new protocol $P'$ for $x$ and $y$ iff $\psi$ is true (for $x$ and $y$).

For $I = (i_1, \dots, i_c) \in [n]^c$, we denote $$\Phi^I_v \equiv \bigwedge_{j=1}^c \left( \bigwedge_{k=1}^{i_j-1} \chi^{j,(k,=)}_v \wedge \chi^{j,(i_j,<)}_v \right).$$
We state what this notation essentially means as a fact.
\begin{fact}
A vertex $v$ is feasible (in the original protocol $P$) for $x$ and $y$ iff there exists $I \in [n]^c$ such that $\Phi^I_v$ is true. We call such an $I$ a \emph{witness} for $v$.
\end{fact}

We first describe a sort of a skeleton of the protocol with equality $P'$ that simulates the given protocol $P$. It is a set of vertices $S$ between which we will later insert trees connecting them. For each sink $\ell \in V$, we put $\ell^{P'}$ into $S$. There, $\ell^{P'}$ is labeled with the conjunction $x_i = 0 \wedge 1 = y_i$ where $i$ is the solution of the KW game corresponding to $\ell$ in $P$. For each inner vertex $v \in V$, we put into $S$ $n^c$ vertices of the form $v(I)$ for $I \in [n]^c$ such that $v(I)$ is labeled with the conjunction $\Phi^I_v$. Observe that $v$ is feasible for $x,y$ in $P$ iff there is a witness $I \in [n]^c$ such that $v(I)$ is feasible in $P'$. An illustration for a vertex $v$ with two children is in \Cref{fig:mainthm_step1}; the question mark corresponds to the part of the protocol which we have not described yet.
The size of the skeleton is upper-bounded by $sn^c$, where $s=|P|$ and $n^c$ is the number of witnesses $I$.

\begin{figure}
    \centering
    \includegraphics[scale=0.8]{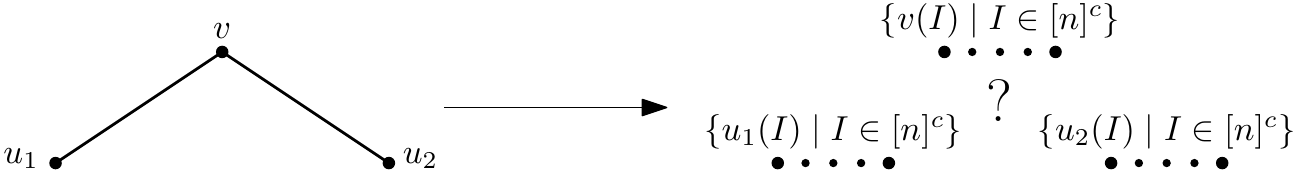}
    \caption{First step of the construction of the protocol $P'$}
    \label{fig:mainthm_step1}
    %\Description{The vertex $v$ with two children $u_1$ and $u_2$ is transformed into three sets of vertices $\{v(I) | I \in [n]^c\}$, $\{u_1(I) | I \in [n]^c\}$, $\{u_2(I) | I \in [n]^c\}$. There is a question mark between vertices corresponding to the part of a protocol to be described later.}
\end{figure}

To construct $P'$ it now remains to construct the trees between the elements of the skeleton and prove their properties. This is done in the following lemma.

\begin{lemma}
\label{l:technical}
Let $\phi$ be a conjunction of bit equalities such that the left-hand side of each equality is a function of $x$ and the right-hand side of each equality is a function of $y$. Let $u_1$, $\dots$, $u_p$ be inner vertices of $P$ and let $\ell_1$, $\dots$, $\ell_{p'}$ be sink vertices of $P$. If the validity of $\phi$ for $x$ and $y$ implies that at least one of the vertices $u_1$, $\dots$, $u_p$ and $\ell_1$, $\dots$, $\ell_{p'}$ is feasible for $x$ and $y$ in $P$, then it is possible to construct a binary protocol $T$ in a form of a tree with the following properties:
\begin{enumerate}
    \item The source of the tree is labeled with $\phi$.
    \item The set of sinks of the tree is $$\left\{ u_{i}(I') \mid i \in [p], I' \in [n]^c \right\} \cup \{ \ell^{P'}_1, \dots, \ell^{P'}_{p'} \}.$$
    \item The size of the tree is $\mathcal{O}(n^{2c(p+p') - 1})$.
    \item The tree satisfies the feasibility condition of Definition~1(b), i.e. if a vertex with two children is feasible for $x$ and $y$, at least one of the children is feasible for $x$ and $y$.
\end{enumerate}
\end{lemma}

Before we prove the lemma, we show that it implies the theorem. First, we apply the lemma with $\phi$ being an empty conjunction and $p' = 0$, $p = 1$, $u_1 = v_0$ (the source of $P$). Second, we apply the lemma for every inner $v \in V$ and $I \in [n]^c$ by setting $\phi \equiv \Phi^I_v$. In this case $u_1, \dots, u_p$ and $\ell_1, \dots, \ell_{p'}$ are the children of $v$. We join all the trees we obtain and identify sources and sinks of these trees when they have the same labels.
For every inner $v \in V$ and $I \in [n]^c$, the vertex $v(I)$ is a source of exactly one tree and a sink of at least one tree. Therefore, the protocol has one source labeled with empty conjunction, which is feasible for any $x$ and $y$, and sinks of the form $\ell^{P'}$ which correspond to solutions for the monotone KW game. The out-degree of every vertex is at most 2 and the feasibility condition of Definition~1(b) is satisfied because the trees satisfy it.

What remains is to estimate the size of the protocol. We apply \Cref{l:technical} once with empty $\phi$ and $p + p' = 1$ and we apply it for every inner $v\in V$ and $I \in [n]^c$. That is at most $s n^c$-times. In these cases it holds $p + p' \leq d$.
The total size is then
$$s n^c \cdot \mathcal{O}(n^{2cd - 1}) + \mathcal{O}(n^{2c - 1}) = \mathcal{O}(sn^{2cd + c - 1}).$$

%%%%%%%%%%%%%%%%%%%%%%%%%%%%%%%%%%%%%%%%%%%%%%%%%%%%%%%%%%%%%%%%%%%%%%%%%

\subsection{Proof of \Cref{l:technical}}

{\bf The protocol as a search procedure.}
We will start by describing the structure of the protocol without the feasibility conditions. We will first assume that all $u_i$s are inner nodes, after that we will say how to modify the protocol if some nodes are sinks.
We will view the protocol as a search procedure. Formally, the search is performed by two communication parties with the help of a referee who tells them which of the available vertices of the protocol is feasible. It is clear, however, that one party (corresponding to the referee) is enough and what it does is constructing a path from the root to a leaf such that all vertices are feasible. We will take the position of the party and imagine that we are performing the search.

The search is structured in 4 levels.

\begin{enumerate}
\item On the highest level we are searching for an index $i\in[d]$ such that $u_i$ is feasible. We search $i$ by systematically testing $i=p,p-1,\dots,1$ until we find $i$ such that $u_i$ is feasible.

\item On the level below the highest we test, for a given $i$, whether the conjunction of inequalities $\bigwedge_{j=1}^c q_{u_i}^j(x)<r_{u_i}^j(y)$ is satisfied. For $j=c,c-1,\dots,1$, we systematically test every inequality $q_{u_i}^j(x)<r_{u_i}^j(y)$ one by one until we either find an inequality that is not satisfied, or verify that all inequalities are satisfied. If we find an inequality that is not satisfied, we proceed to the next conjunction. If all conjunctions are satisfied, the search is completed and we are at a leaf of $T$.

\item On the next level below we test inequalities $q_{u_i}^j(x)<r_{u_i}^j(y)$.
Let $a=q_{u_i}^j(x), b=r_{u_i}^j(y)$. To test $a<b$?, we compare the bits of $a$ and $b$ starting from the most significant one and proceeding to less significant ones. The protocol is based on formula \cref{ineq}, which means that we continue as long as the bits are equal. So we have three cases:
 \begin{enumerate}
 \item If for the currently tested bit $k$, $a[k]=b[k]$, the search goes on if $k>1$. If $k=1$ then $a=b$, so the players know that $u_i$ is not feasible and we go to test next node $u_{i-1}$.
 \item If we find $k$ such that $a[k]<b[k]$, we know that $q_{u_i}^j(x)<r_{u_i}^j(y)$  and we start testing the next inequality if $j>1$, or finish testing if $j=1$, because $u_i$ is feasible.
 \item If we find $k$ such that $a[k]>b[k]$, we know that $q_{u_i}^j(x)>r_{u_i}^j(y)$. In this case we know that $u_i$ is not feasible and we go to test next conjunction. %node $u_{i-1}$.
 \end{enumerate}

\item On the lowest level we need to determine which of the three possibilities $a[k]<b[k],a[k]=b[k],a[k]>b[k]$ holds true. Since the simulation is by a protocol of degree 2, we first decide whether $a[k]=b[k]$ or $a[k]\neq b[k]$ and in the second case we decide whether $a[k]<b[k]$ or $a[k]>b[k]$. If we are  testing the first conjunction, the feasibility conditions ensure that that none of $u_i$ $i>1$ is feasible, so $u_1$ must be feasible. Which means that all inequalities in the conjunction are true and thus
we only need to test whether $a[k]=b[k]$ or $a[k]<b[k]$. (This is not essential, it only makes the tree a little smaller than it would be if we kept the superfluous test for $a[k]>b[k]$.) 
\end{enumerate}

This search procedure can be naturally represented as a directed acyclic graph or a tree but we also need to label the vertices with feasibility conditions. Therefore have to use a tree, because the feasibility conditions depend on the history of the search. To estimate the size of the tree we need a more explicit description, which we will postpone to the next section.

\medskip
{\bf Feasibility conditions.}
Now we describe the feasibility conditions used in the search procedure. The feasibility conditions are given by two strings, one depending on $x$, the other depending on $y$. At the root of the tree we have two strings that make the condition of $v(I)$. In each step we extend the strings by one or two more bits. Since the protocol has the form of a tree, we do not have to forget any bits. (We could forget some bits at some stages of the protocol and thus make the protocol DAG-like and slightly smaller, but the gain would not be significant and it would complicate the proof.) The added bits are determined by what the protocol tests.

It will be more convenient to view the feasibility conditions as conjunctions of $\Phi$ and some elementary equality conditions  $\chi^{j,(k,=)}_v,\chi^{j,(k,\neq)}_v,\chi^{j,(k,<)}_v,\chi^{j,(k,>)}_v$. At each step of the procedure we add a new elementary term.

\begin{enumerate}
\item When we test $a[k]=b[k]$? and $w$ (respectively $z$) is the child where $a[k]=b[k]$ (respectively $a[k]\neq b[k]$), the feasibility condition is extended with $\chi^{j,(k,=)}_{u_i}$ (respectively with $\chi^{j,(k,\neq)}_{u_i}$).
  % strings are extended with $a[k]$ and $b[k]$ (with $a[k]$ and $1-b[k]$).
\item When we test whether $a[k]<b[k]$ or $a[k]>b[k]$ and $w$ (respectively $z$) is the child where $a[k]<b[k]$ (respectively $a[k]> b[k]$), the feasibility condition is extended with $\chi^{j,(k,<)}_{u_i}$ (respectively with $\chi^{j,(k,>)}_{u_i}$).
%$w$ ($z$) is the child where $a[k]<b[k]$ ($a[k]>b[k]$), then the strings are extended with $a[k]1$ and $0b[k]$ (with $a[k]0$ and $1b[k]$).
\end{enumerate}
This, clearly, satisfy the condition of the feasibility predicates of Definition~1(b), because in the first case always $a[k]=b[k]$ or $a[k]\neq b[k]$ holds true, and in the second case $a[k]<b[k]$ or $a[k]>b[k]$ holds true, because in the parent node we have  $a[k]=1-b[k]$.

At the leaves $u_i(I)$ of the protocol, all bits terms are removed except for those that witness the feasibility of $u_i$ in $P$. Thus we get $\Phi^I_{u_i}$.

This determines the feasibility conditions in the protocol. For the sake of higher clarity, we will describe explicitly the conditions in particular steps, but we will only say what is added at these steps to the previous stings. We will consider the levels of the protocol in the inverse order.

\begin{itemize}
    
\item We have already said what happens on the lowest level.

\item When $a<b$? is tested and after the $k$th bit has been tested without deciding the inequality, the last $k$ terms in the feasibility condition express that $a[n]a[n-1]\dots a[n-k+1]=b[n]b[n-1]\dots b[n-k+1]$. This guarantees that if $a\neq b$, the two numbers differ at a lower bit. After testing the next bit (if there is any) the condition will be extended as defined above.

\item When testing a conjunction $a^1<b^1\wedge a^2<b^2\wedge\dots \wedge a^c<b^c$ and the last $j$ inequalities were tested positive, the strings of the feasibility conditions will contain witnesses of these inequalities. When an inequality $a^j<b^j$ is tested negatively, the witness for $a^j\geq b^j$ is added; there are two versions, one for $a^j=b^j$ and one for $a^j> b^j$. (At this point we could remove the witnesses of $a^c<b^c,\dots,a^{c-j+1}<b^{c-j+1}$, but we do not do it.)

\item When testing the feasibility of $u_i$ for $i>1$,
the feasibility condition contains formulas that witness that all $u_{i+1},\dots,u_p$ are not feasible.
If $u_i$ is tested positively, the communication proceeds to the leaf $u_i(I)$ where $I$ consists of the witnesses for the members of the conjunction. Otherwise the a witness of the failure is added and the communication proceeds with testing $u_{i-1}$.
\end{itemize}

\noindent The property of the feasibility conditions in Definition~1(b) is satisfied because:

\begin{enumerate}
\item When testing conjunctions for $u_i$ for $i<d$, the protocol considers all 3 possibilities $a[k]<b[k],a[k]=b[k],a[k]>b[k]$ one of which must be true.
\item In the first conjunction we only test whether $a[k]<b[k]$ or $a[k]=b[k]$, because the feasibility conditions contain terms witnessing the feasibility of $v$ and non-feasibility $u_2,\dots,u_{d}$ in $P$, which imply the feasibility of $u_1$. The last bit is not tested, because if the communication reaches this node the feasibility of $v$, non-feasibility $u_2,\dots,u_{d}$ and $a[n]\dots a[2]=b[n]\dots b[2]$ force $a[1]<b[1]$. 
\end{enumerate}

Note that in spite of the fact that the first conjunction is always tested positively if the communication reaches this stage, we have to ``test'' it. In this case it is not literally a test, but a search for a witness $I$ that guarantees that the conjunction is true.

\subsubsection{Estimating the size of T}

We can assume $p'=0$, i.e., there are no sinks among $u_i$s, because if there are some, the tree is smaller.

We construct the binary tree $T$ recursively by defining its subtrees $T^\psi_{i,j,k}$. The bottom indices determine the stage of the search: $i$ corresponds to vertex $u_i$, $j$ to the element of the tested conjunction, $k$ to the tested bit. The superscript $\psi$ denotes the feasibility condition in the root of $T^\psi_{i,j,k}$. If we disregard the feasibility conditions, then for fixed $i,j,k$, all trees $T^\psi_{i,j,k}$ are isomorphic. The feasibility conditions $\psi$, however, are different, because the stage $(i,j,k)$ of the communication may be reached in different ways and the history is encoded in $\psi$. We have already defined the feasibility conditions above. Here we 
will define the feasibility conditions $\psi$ for the trees $T^\psi_{i,j,k}$ only for the sake of completeness; they are not needed for estimating the size of $T$.
%When defining the conditions recursively we will use conjunctions, which makes the definitions simpler. (Recall that a conjuction of several equalities can be stated as one equality.)

Trees $T^\psi_{i,j,k}$ are defined recursively using the lexicographic order starting with $i=1,j=1,k=2$. (Recall that we do not test the pair of least bits in the first inequality of the first conjunction, because their values are forced to 0 and 1; therefore we start with $k=2$, but in other instances of $i$ and $j$ we do have $k=1$.) The whole tree is $T = T^\phi_{p,c,n}$.

The tree $T^\psi_{i,j,k}$ has the root labeled with the formula $\psi$ that is a conjunction of $\phi$ (the feasibility condition of $v$ in $P$) and the elementary terms $\chi^{j,(k,=)}_v,\chi^{j,(k,\neq)}_v,\chi^{j,(k,<)}_v,\chi^{j,(k,>)}_v$ added on the path to its root.

Case 1: $i = 1$. For $j \in [c]$ and $k \in \{2, \dots, n\}$. The tree $T^\psi_{1,j,k}$ has a root labeled with $\psi$. There are two subtrees under the root, named $A^=$ and $A^<$. We define:
\begin{align*}
\psi^= &\equiv \psi \wedge \chi^{j,(n-k+1,=)}_{u_1}
\\
\psi^< &\equiv \psi \wedge \chi^{j,(n-k+1,<)}_{u_1}
\end{align*}
The particular choices for $A^=$ and $A^<$ are in the table. The superscript for $T$ is omitted; it is $\psi^=$ in the column $A^=$ and $\psi^<$ in the column $A^<$. When there is $u_1(I_1)$ or $u_2(I_2)$ in the table, the subtree consists of a single vertex $u_1(I_1)$ or $u_2(I_2)$. 

\begin{tabular}{lcccc}
Case &  $j$ & $k$ & $A^=$ & $A^<$
\\
\hline
1.1 & $=1$ & $=2$ & $u_1(I_1)$ & $u_1(I_2)$ \\
1.2 & $=1$ & $>2$ & $T_{1,1,k-1}$ & $u_1(I_2)$ \\
1.3 & $>1$ & $=2$ & $T_{1,j-1,n}$ & $T_{1,j-1,n}$ \\
1.4 & $>1$ & $>2$ & $T_{1,j,k-1}$ & $T_{1,j-1,n}$
\end{tabular}

\noindent The meaning of this table is as follows. In Case~1.1. this is the last step of the search procedure, hence the tree is connected to $u_1(I_1)$ and $u_1(I_2)$. The conditions $I_1$ and $I_2$ are witnesses of the first conjunction being true; they differ only in the last bits. In Case~1.2 the first conjunction is being tested and the bits of the numbers that are tested are $>2$. Hence the testing can either continue, if the bits are same (column $A^=$), or finish if the first is less then the second (column $A^<$). In Case 1.3 the second least significant bit is tested in the $j$th conjunction. Hence testing proceeds to the next conjunction (the $j-1$st) and the subtrees representing this search differ only in their feasibility conditions. In Case~1.4 a higher bit is tested in the $j$th conjunction. If the tested bits are equal, testing of the inequality continues, otherwise next inequality is tested.

\begin{figure}
    \centering
    \includegraphics[scale=1]{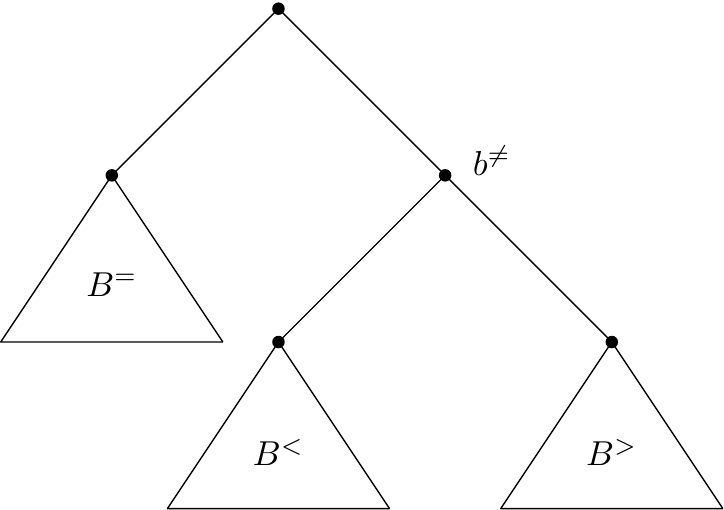}
    \caption{Case 2 of the construction}
    \label{fig:tecko}
    %\Description{A tree with a root (the root has no name). The left child of the root is a subtree $B^=$. The right child is a vertex $b^\neq$, which has two subtrees---$B^<$ and $B^>$.}
\end{figure}

Case 2: $i > 1$. The difference is that now we have to test all three possibilities  $a[k]<b[k],a[k]=b[k],a[k]>b[k]$ for the tested $k$th bit.

For $i \in \{2, \dots, p\}$, $j \in [c]$ and $k \in [n]$, the tree $T^\psi_{i,j,k}$ has again a root labeled with $\psi$. The left subtree is $B^=$, the right subtree has a root $b^{\neq}$ which has two subtrees $B^<$ and $B^>$. See  \Cref{fig:tecko}.
We define:
\begin{align*}
\psi^{\neq} &\equiv \psi \wedge \chi^{j,(n-k+1,\neq)}_{u_i}
\\
\psi^= &\equiv \psi \wedge \chi^{j,(n-k+1,=)}_{u_i}
\\
\psi^< &\equiv \psi \wedge \chi^{j,(n-k+1,<)}_{u_i}
\\
\psi^> &\equiv \psi \wedge \chi^{j,(n-k+1,>)}_{u_i}
\end{align*}
The label of $b^{\neq}$ is $\psi^{\neq}$.
Again, we omit the superscript for $T$ which is $\psi^t$ in the column $B^t$ for $t \in \{ <, >, = \}$.

\begin{tabular}{lccccc}
Case & $j$ & $k$ & $B^=$ & $B^<$ & $B^>$
\\
\hline
2.1 & $=1$ & $=1$ & $T_{i-1,c,n}$ & $u_i(I_2)$ & $T_{i-1,c,n}$\\
2.2 & $=1$ & $>1$ & $T_{i,1,k-1}$ & $u_i(I_2)$ & $T_{i-1,c,n}$\\
2.3 & $>1$ & $=1$ & $T_{i-1,c,n}$ & $T_{i,j-1,n}$ & $T_{i-1,c,n}$\\
2.4 & $>1$ & $>1$ & $T_{i,j,k-1}$ & $T_{i,j-1,n}$ & $T_{i-1,c,n}$\\
\end{tabular}

We leave to the reader to figure out the meaning of these cases, because it is very similar to the case of~$i=1$. We now compute the size of $|T_{1,j,k}|$.

\begin{claim}
\label{claim:Tsizecase1}
For $j \in [c]$ and $k \in \{2, \dots, n\}$
$$\lvert T_{1,j,k} \rvert = 2kn^{j-1} - 1.$$
\end{claim}

\begin{proof}
We use induction and verify cases 1.1 -- 1.4:\\
Case 1.1: $$\lvert T_{1,1,2} \rvert = 3$$
Case 1.2: For $k \in \{3, \dots, n\}$:
$$\lvert T_{1,1,k} \rvert = 2 + \lvert T_{1,1,k-1} \rvert = 2 + 2(k-1)
-1 = 2k - 1$$
Case 1.3: For $j \in \{2, \dots, c\}$:
$$\lvert T_{1,j,2} \rvert = 1 + 2\lvert T_{1,j-1,n} \rvert = 1 +
2\cdot (2n\cdot n^{j-2} - 1) = 4n^{j-1} - 1$$
Case 1.4: For $j \in \{2, \dots, c\}$ and $k \in \{3, \dots, n\}$:
$$\lvert T_{1,j,k} \rvert = 1 + \lvert T_{i,j,k-1} \rvert + \lvert
T_{i,j-1,n} \rvert = 1 + 2(k-1)n^{j-1} - 1 + 2n\cdot n^{j-2} - 1 = 2kn^{j-1} - 1$$
\end{proof}

In the more general cases, it would be too tedious to write down the exact formulas and we instead use approximate bounds.

\begin{claim}
\label{claim:Tsize}
For $n \geq 10$:
$$\lvert T_{p,c,n} \rvert \leq 2n^{2pc - 1}$$
\end{claim}

\begin{proof}
We prove by induction that for $i \in [p], j \in [c], k \in [n]$
(except for the case $i = 1$ and $k = 1$ when $T_{i,j,k}$ is undefined):
$$\lvert T_{i,j,k} \rvert \leq 2kn^{2(i - 1)c + 2(j - 1)}$$

This holds for the case $i = 1$ because the formula from \Cref{claim:Tsizecase1} satisfies $$2kn^{j-1}-1 \leq 2kn^{j-1} \leq 2kn^{2(j - 1)}.$$

We consider cases 2.1 -- 2.4:\\
Case 2.1: For $i \in \{2, \dots, p\}$: $$\lvert T_{i,1,1} \rvert = 3 + 2\lvert T_{i-1,c,n} \rvert \leq 3 +
2\cdot 2n\cdot n^{2(i - 2)c+2(c - 1)} \leq 7n^{2(i - 1)c - 1} \leq 2n^{2(i-1)c}$$
Case 2.2: For $i \in \{2, \dots, p\}$ and $k \in \{2, \dots, n\}$:
\begin{align*}\lvert T_{i,1,k} \rvert &= 3 + \lvert T_{i,1,k-1} \rvert + \lvert
T_{i-1,c,n} \rvert \leq 3 + 2(k-1)n^{2(i-1)c} + 2n\cdot n^{2(i-2)c + 2(c-1)}\\
& = 2(k-1)n^{2(i-1)c} + (2 + 2n^{2(i-1)c - 1}) \leq
2k n^{2(i-1)c}\end{align*}
Case 2.3:  For $i \in \{2, \dots, p\}$ and $j \in \{2, \dots, c\}$:
\begin{align*}\lvert T_{i,j,1} \rvert &= 2 + 2\lvert T_{i-1,c,n} \rvert + \lvert T_{i,j-1,n} \rvert \\
& \leq 2 + 2\cdot 2n\cdot n^{2(i-2)c + 2(c-1)} + 2n\cdot n^{2(i-1)c + 2(j-2)}\\
& = 2 + 4n^{2(i-1)c - 1} + 2n^{2(i-1)c + 2(j - 1) - 1} \leq 2n^{2(i-1)c +
2(j-1)}\end{align*}
Case 2.4: For $i \in \{2, \dots, p\}$, $j \in \{2, \dots, c\}$ and $k
\in \{2, \dots, n\}$:
\begin{align*}\lvert T_{i,j,k} \rvert &= 2 + \lvert T_{i,j,k-1} \rvert + \lvert
T_{i,j-1,n} \rvert + \lvert T_{i-1,c,n} \rvert\\
	& \leq 2 + 2(k-1)\cdot n^{2(i-1)c + 2(j-1)} \\& + 2n\cdot n^{2(i-1)c + 2(j-2)} + 2n\cdot
n^{2(i-2)c + 2(c-1)} \\
	& = 2(k-1) n^{2(i-1)c + 2(j-1)} + 2\\& + 2n^{2(i-1)c + 2(j - 1) - 1} +
2n^{2(i-1)c - 1} \\
&\leq  2k n^{2(i-1)c + 2(j-1)}
\end{align*}

All the cases are checked and the statement of the claim follows.
\end{proof}

As noted above, if we have $p'$ leaves on top of the $p$ inner vertices $u_i$, the size of $T$ is at most (in fact less, except for trivial cases) the size of 
$T_{n,p+p',c}$. Hence the size of $T$ can be bounded by $\mathcal{O}(n^{2(p+p')c-1})$.

This finishes the proof of \Cref{l:technical} and also the proof of \Cref{thm:main}.

\section{A potential application of protocols with equality}
\label{sec:applications}

We prove in this section how one could turn lower bounds for protocols from this paper to lower bounds for the proof system R(LIN). The idea of using with equality for R(LIN) is due to D. Sokolov (personal communication with P. Pudlák). This result has not been published, therefore, we present it here for the sake of completeness.

We define R(LIN) in the spirit of Krajíček's book~\cite{krajicekproofcomplexity}. However, we use it only as a~proof system for DNF tautologies (i.e. a refutation system for CNF formulas). Therefore, the Boolean axioms and the weakening rule can be omitted.

\begin{definition}[R(LIN)]
The system R(LIN) is an extension of resolution. A~clause in the system R(LIN)
	is a~set $C = \{f_1, \dots, f_k\}$ where all elements $f_i \in \mathbf{F}_2[x_1, \dots, x_n]$ are
linear polynomials. As it is usual, the value 0 is interpreted as false
	and the value 1 as true.  Then the~clause $C$ is true for $\overline{a} \in \{0,1\}^n$ if and only if the Boolean formula $\bigvee_{i=1}^k f_i$ is true.

Let $\phi$ be a~CNF formula with the set of clauses $S$. Replacing each negative literal $\neg x_i$ with $1+x_i$, we obtain an equivalent set of clauses $S'$ for R(LIN).
	A R(LIN)-refutation of $\phi$ is then the derivation of the empty clause from the clauses $S'$ using the following two rules.

The contraction rule:
\begin{prooftree}
\AxiomC{$C,0$}
\UnaryInfC{$C$}
\end{prooftree}

The addition rule:
\begin{prooftree}
\AxiomC{$C,g$}
\AxiomC{$D,h$}
\BinaryInfC{$C,D,g+h+1$}
\end{prooftree}
\end{definition}

To translate the lower bounds from protocols to proofs, we need to 
convert a function $f$ which is hard to solve by a~protocol to a formula which is hard to refute.
The following lemma is essentially Lemma 13.5.1 from Krajíček's book \cite{krajicekproofcomplexity}.

\begin{lemma}
\label{l:functiontoformula}
Let $f\colon \{0,1\}^* \rightarrow \{0,1\}$ be a partial monotone function such that $U = f^{-1}\{0\}$ and $V = f^{-1}\{1\}$ are NP-sets.
Then for each $n \geq 1$, there exists a CNF formula $\phi(\overline{x}, \overline{y}) \wedge \psi(\overline{x}, \overline{z})$ such that:
\begin{enumerate}
\item variables from $\overline{y}$ and $\overline{z}$ are disjoint,
\item variables from $\overline{x}$ occur only negatively in $\phi(\overline{x}, \overline{y})$,
\item variables from $\overline{x}$ occur only positively in $\psi(\overline{x}, \overline{z})$,
\item \label{someitem} in each clause of $\psi(\overline{x}, \overline{z})$, there is at most one occurrence of a variable from $\overline{x}$,
\item the size of $\phi(\overline{x}, \overline{y}) \wedge \psi(\overline{x}, \overline{z})$  is polynomial in $n$,
\item $U \cap \{0,1\}^n = \left\{ \overline{a} \in \{0,1\}^n \middle| \phi(\overline{a},\overline{y}) \in SAT \right\}$,
\item $V \cap \{0,1\}^n = \left\{ \overline{b} \in \{0,1\}^n \middle| \psi(\overline{b},\overline{z}) \in SAT \right\}$.
\end{enumerate}
\end{lemma}

\begin{proof}
Except for \Cref{someitem}, this is Lemma 13.5.1 from \cite{krajicekproofcomplexity}.

	We replace each clause of $\phi(\overline{x}, \overline{z})$ by an equivalent set of clauses.
	For example, clause $\{x_1, x_2, x_3, z_1, \neg z_2, z_3\}$ is transformed into an equivalent set of clauses $\{z_1, \neg z_2, z_3, x_1, z_4\}$, $\{\neg z_4, x_2, z_5\}$, $\{\neg z_5, x_3\}$
	where $z_4, z_5$ are new variables.  Generally: Consider a clause of
	the form $\{\ell_1, \dots, \ell_k, \ell'_1, \dots, \ell'_m\}$ where
	$\ell_1, \dots, \ell_k$ are literals from $\overline{x}$ and $\ell'_1,
	\dots, \ell'_m$ are literals from $\overline{z}$. Such a~clause is
	replaced by $k$ clauses: $\{\ell'_1, \dots, \ell'_m, \ell_1,
	\ell''_1\}$, $\{\neg \ell''_1, \ell_2, \ell''_2\}$, \dots,$\{\neg
	\ell''_{k-2}, \ell_{k-1}, \ell''_{k-1}\}$, $\{\neg \ell''_{k-1}, \ell_k\}$ where $\ell''_1, \dots, \ell''_{k-1}$ are  new variables put into $\overline{z}$.
\end{proof}

Observe that $\phi(\overline{x}, \overline{y}) \wedge \psi(\overline{x}, \overline{z})$ is unsatisfiable because $U \cap V = \emptyset$.

We are now ready to state the theorem for R(LIN) and protocols with equality.

\begin{theorem}
Let $f\colon \{0,1\}^* \rightarrow \{0,1\}$ be a monotone partial function such that $U = f^{-1}(0)$ and $V = f^{-1}(1)$ are NP-sets.
	Then it holds for all $n$: If $s(n)$ is the minimum size of a DAG
communication protocol of degree 2 with equality solving $f_n =
\left.f\right|_{\{0,1\}^n}$, then there is a formula $\alpha^f_n$ of size
$n^{O(1)}$ such that the minimum number of lines of a R(LIN)-refutation of
$\alpha^f_n$ is $s(n)$.
\end{theorem}

\begin{proof}
Let us fix $n$. We use the formula $\phi(\overline{x}, \overline{y}) \wedge \psi(\overline{x}, \overline{z})$ from \Cref{l:functiontoformula} as $\alpha^f_n$.

For the sake of contradiction, we assume that there is a R(LIN)-refutation $d$
of $\phi(\overline{x}, \overline{y}) \wedge \psi(\overline{x}, \overline{z})$ whose number
of lines is $s < s(n)$.  We construct a protocol $P$ of degree 2 with equality of size at most $s$ solving
$f_n$. The underlying graph of the protocol $P$ is the graph of the proof $d$ with the direction of edges reversed.  For each $\overline{a} \in f^{-1}\{0\}$, we fix $\overline{y_a}$ such that
$\phi(\overline{a},\overline{y_a})$ is true. Analogously, for each 	$\overline{b} \in f^{-1}\{1\}$,
we fix $\overline{z_b}$ such that $\psi(\overline{b},\overline{z_b})$ is true.

Let us consider a vertex $v$ with the clause $f_1, \dots, f_k$. The
functions $q_v$, $r_v$ will be set so that the vertex $v$ is
feasible for $\overline{a}, \overline{b}$ iff the assignment $\overline{a}, \overline{y_a},
\overline{z_b}$ falsifies $f_1, \dots, f_k$.  For each $i$, we split $f_i = f_i^0
+ f_i^1$ so that $f_i^0$ contains only variables from $\overline{x},
\overline{y}$ and the constant 1 if present, while $f_i^1$ contains
only variables from $\overline{z}$.  It is then enough to set
$q_v(\overline{a})$ to the sequence of bits $v_1v_2\dots v_k$ where $v_i$ is the
value of $f^0_i$ under the assignment $\overline{a}, \overline{y_a}$ and to set $r_v(\overline{b})$ to
the sequence of bits $w_1w_2\dots w_k$ where $w_i$ is the value of
$f^1_i$ under the assignment $\overline{z_b}$.

First, observe that the source of the graph is always feasible because the empty clause is falsified by any assignment.
Next, we observe that if a vertex $v$ is feasible for $\overline{a}$ and $\overline{b}$, then at
least one of the sons is feasible for $\overline{a}$ and $\overline{b}$, too. This is trivial for the contraction rule of R(LIN). Let us consider the vertex $v$ with the clause
$C,D,f+g+1$ and its two sons with clauses $C,f$ and $D,g$. Because $f+g+1$ is false under the assignment $\overline{a}, \overline{y_a}, \overline{z_b}$, at least one of $f$, $g$ must be false under the same assignment. The clauses $C$, $D$ are trivially false under this assignment. Therefore at least one of the sons is indeed feasible.

Let us consider a vertex $v$ such that all of its descendants come
exclusively from the clauses of $\phi(\overline{x},\overline{y})$. If the clause of $v$ is
falsified by the assignment $\overline{a}, \overline{y_a}$, it contradicts the fact that
$\phi(\overline{a}, \overline{y_a})$ is true. Therefore the vertex $v$ is never feasible. We remove all such vertices. The only remaining thing is to verify that the sinks give us a correct solution to the KW game. Each sink $\ell$ must have a clause $C$ from $\psi(\overline{x},\overline{z})$. We claim that the sink is feasible if and only if $a_i = 0 \wedge b_i = 1$ where $x_i$ is the (only) positive literal of $x$ appearing in the clause $C$. Feasibility of $\ell$ implies that all literals in $C$ are falsified by the assignment $\overline{a}, \overline{z_b}$. But $C$ is true under the assignment $\overline{b},\overline{z_b}$. Therefore, $a_i = 0 \wedge b_i = 1$ as required.

We have constructed a protocol solving $f_n$ of size at most $s$. Because $s < s(n)$, this is a contradiction.
\end{proof}

\section*{Acknowledgements}
I want to thank Pavel Pudlák for his guidance through this work. I am also grateful to Susanna F. de Rezende for a discussion and comments. This work was supported by the GAUK project 496119.

\printbibliography

\end{document}